\documentclass[11pt]{article}    

\usepackage[margin=0.75in]{geometry} 
\usepackage{amsmath,amssymb,amsthm}

\newtheorem{theorem}{Theorem}[section]
\newtheorem{proposition}[theorem]{Proposition}

\newtheorem{corollary}[theorem]{Corollary}

\newtheorem{remark}{Remark}

\newcommand{\al}{\alpha}
\newcommand{\bt}{\beta}
\newcommand{\la}{\lambda}

\newcommand{\s}{\sigma}

\newcommand{\be}{\begin{equation}}
\newcommand{\ee}{\end{equation}}
\newcommand{\bea}{\begin{eqnarray}}
\newcommand{\eea}{\end{eqnarray}}
\newcommand{\no}{\nonumber}

\numberwithin{equation}{section}
\begin{document}

\title{Hankel Determinant and Orthogonal Polynomials for a Perturbed Gaussian Weight: from Finite $n$ to Large $n$ Asymptotics}
\author{Chao Min\thanks{School of Mathematical Sciences, Huaqiao University, Quanzhou 362021, China; e-mail: chaomin@hqu.edu.cn}\: and Yang Chen\thanks{Department of Mathematics, Faculty of Science and Technology, University of Macau, Macau, China; e-mail:\qquad\qquad\qquad\qquad yangbrookchen@yahoo.co.uk}}


\date{\today}
\maketitle
\begin{abstract}
We study the monic polynomials orthogonal with respect to a symmetric perturbed Gaussian weight
$$
w(x;t):=\mathrm{e}^{-x^2}\left(1+t\: x^2\right)^\lambda,\qquad x\in \mathbb{R},
$$
where $t> 0,\;\lambda\in \mathbb{R}$. This weight is related to the single-user MIMO systems in information theory. It is shown that the recurrence coefficient $\beta_n(t)$ is related to a particular Painlev\'{e} V transcendent, and the sub-leading coefficient $\mathrm{p}(n,t)$ satisfies the Jimbo-Miwa-Okamoto $\sigma$-form of the Painlev\'{e} V equation. Furthermore, we derive the second-order difference equations satisfied by $\beta_n(t)$ and $\mathrm{p}(n,t)$, respectively. This enables us to obtain the large $n$ full asymptotic expansions for $\beta_n(t)$ and $\mathrm{p}(n,t)$ with the aid of Dyson's Coulomb fluid approach. We also consider the Hankel determinant $D_n(t)$, generated by the perturbed Gaussian weight. It is found that $H_n(t)$, a quantity allied to the logarithmic derivative of $D_n(t)$, can be expressed in terms of $\beta_n(t)$ and $\mathrm{p}(n,t)$. Based on this result, we obtain the large $n$ asymptotic expansion of $H_n(t)$ and then that of the Hankel determinant $D_n(t)$.
\end{abstract}

$\mathbf{Keywords}$: Orthogonal polynomials; Hankel determinant; Perturbed Gaussian weight;

Painlev\'{e} V; Coulomb fluid; Asymptotics.

$\mathbf{Mathematics\:\: Subject\:\: Classification\:\: 2020}$: 42C05, 33E17, 41A60.

\section{Introduction}
Given a positive Borel measure $\mu$ on the real line, where all the moments
$$
\mu_{j}:=\int_{\mathbb{R}}x^j d\mu(x),\qquad j=0, 1, 2,\ldots
$$
exist, it is well known that one can define the monic orthogonal polynomials $P_n(x), n=0, 1, 2,\ldots$ such that
$$
\int_{\mathbb{R}}P_m(x)P_n(x)d\mu(x)=h_n\delta_{mn},
$$
where $h_n>0$ and $\delta_{mn}$ denotes the Kronecker delta.

An important property of the orthogonal polynomials is the three-term recurrence relation
$$
xP_{n}(x)=P_{n+1}(x)+\al_n P_n(x)+\beta_{n}P_{n-1}(x),
$$
supplemented by the initial conditions
$$
P_0(x):=1,\qquad P_{-1}(x):=0.
$$
The recurrence coefficients $\al_n$ and $\bt_n$ have the integral representations:
$$
\al_n=\frac{1}{h_n}\int_{\mathbb{R}}xP_n^2(x)d\mu(x),
$$
$$
\bt_n=\frac{1}{h_{n-1}}\int_{\mathbb{R}}xP_n(x)P_{n-1}(x)d\mu(x).
$$
See \cite{Chihara,Ismail,Szego} for information on orthogonal polynomials.

For classical orthogonal polynomials, namely the Hermite, Laguerre and Jacobi polynomials, the measure $\mu(x)$ is absolutely continuous and we write $d\mu(x)=w(x)dx$. Here $w(x)$ is the weight function corresponding to the normal, gamma and beta distributions, respectively. Furthermore, the weight functions of classical orthogonal polynomials satisfy the Pearson equation (see \cite{AN})
\be\label{pe}
\frac{d}{dx}(\varrho(x)w(x))=\tau(x)w(x),
\ee
where $\varrho(x)$ is a polynomial of degree $\leq 2$ and $\tau(x)$ a polynomial of degree 1.
Semi-classical orthogonal polynomials are defined to have a weight function $w(x)$ that satisfies the Pearson equation (\ref{pe}) where $\varrho(x)$ and $\tau(x)$ are polynomials with deg $\varrho>2$ or deg $\tau\neq 1$; see \cite[Section 1.1.1]{VanAssche}.

In the present paper, we consider the monic orthogonal polynomials with respect to a symmetric perturbed Gaussian weight
\be\label{weight}
w(x)=w(x;t):=\mathrm{e}^{-x^2}\left(1+t\:x^2\right)^\la,\qquad x\in \mathbb{R}
\ee
with $t> 0,\;\la\in \mathbb{R}$. It is easy to see that (\ref{weight}) is a semi-classical weight since it satisfies the Pearson equation (\ref{pe}) with
$$
\varrho(x)=1+t\:x^2,\qquad \tau(x)=2x \left[ (\la+1)  t-1\right]-2 t\: x^3.
$$
It should be pointed out that our weight (\ref{weight}) is closely related to the weight $w(x)=x^{\al}\mathrm{e}^{-x}(x+t)^{\lambda},\;x\in \mathbb{R}^{+}, \al>-1, t>0$. The latter weight plays an important role in the single-user MIMO systems \cite{BC,Chen2012}.

Our study is motivated in part by the fact that the recurrence coefficients of semi-classical orthogonal polynomials are usually related to the solutions of Painlev\'{e} equations. This fact was first observed by Fokas et al. \cite{Fokas}, and Magnus \cite{Magnus,Magnus1} made a great contribution in this field. Other classical examples are mentioned below. For example, Basor et al. \cite{BCE} showed that the recurrence coefficients of the monic orthogonal polynomials with weight $w(x)=(1-x)^\al(1+x)^{\bt}\mathrm{e}^{-tx}, x\in [-1,1], \al, \bt>-1, t\in\mathbb{R}$, can be expressed in terms of the solution of a particular Painlev\'{e} V; Filipuk et al. \cite{Filipuk} and Clarkson et al. \cite{Clarkson} discussed the relation between recurrence coefficients of the semi-classical Laguerre polynomials with weight $w(x)=x^\al \mathrm{e}^{-x^2+tx}, x\in \mathbb{R}^{+}, \al>-1, t\in \mathbb{R}$ and the Painlev\'{e} IV equation. See also \cite{BC,Dai,FW2007,Min2020,VanAssche} for reference.


Let $P_{n}(x;t),\; n=0,1,2,\ldots$ be the monic polynomials of degree $n$ orthogonal with respect to the weight (\ref{weight}), i.e.,
\be\label{or}
\int_{-\infty}^{\infty}P_{m}(x;t)P_{n}(x;t)w(x;t)dx=h_{n}(t)\delta_{mn},\qquad m, n=0,1,2,\ldots,
\ee
where $P_{n}(x;t)$ has the following expansion
\be\label{expan}
P_{n}(x;t)=x^{n}+\mathrm{p}(n,t)x^{n-2}+\cdots,
\ee
since the weight $w(x;t)$ is even \cite[p. 21]{Chihara}.
Here $\mathrm{p}(n,t)$ denotes the coefficient of $x^{n-2}$, and the initial values of $\mathrm{p}(n,t)$ are set to be $\mathrm{p}(0,t)=0,\; \mathrm{p}(1,t)=0$.

The three-term recurrence relation for our orthogonal polynomials now reads \cite[p. 18-21]{Chihara}
\be\label{rr}
xP_{n}(x;t)=P_{n+1}(x;t)+\beta_{n}(t)P_{n-1}(x;t),
\ee
with the initial conditions
$$
P_{0}(x;t)=1,\qquad P_{-1}(x;t)=0.
$$
From (\ref{or}) and (\ref{rr}) we have
\be\label{be2}
\beta_{n}(t)=\frac{h_{n}(t)}{h_{n-1}(t)}.
\ee
Moreover, the combination of (\ref{expan}) and (\ref{rr}) gives
\be\label{be1}
\beta_{n}(t)=\mathrm{p}(n,t)-\mathrm{p}(n+1,t),
\ee
and then $\bt_0(t)=0$. Taking a telescopic sum of (\ref{be1}), we find
\be\label{sum}
\sum_{j=0}^{n-1}\beta_{j}(t)=-\mathrm{p}(n,t).
\ee

The Hankel determinant generated by the weight (\ref{weight}) is defined by
$$
D_{n}(t):=\det(\mu_{j+k}(t))_{j,k=0}^{n-1}=\begin{vmatrix}
\mu_{0}(t)&\mu_{1}(t)&\cdots&\mu_{n-1}(t)\\
\mu_{1}(t)&\mu_{2}(t)&\cdots&\mu_{n}(t)\\
\vdots&\vdots&&\vdots\\
\mu_{n-1}(t)&\mu_{n}(t)&\cdots&\mu_{2n-2}(t)
\end{vmatrix}
,
$$
where $\mu_{j}(t)$ is the $j$th moment given by
\bea\label{moment}
\mu_{j}(t)&=&\int_{-\infty}^{\infty}x^{j}w(x;t)dx\nonumber\\
&=&\left\{
\begin{aligned}
&0,&j=1,3,5,\ldots;\\
&t^{-\frac{j+1}{2}}\Gamma\bigg(\frac{j+1}{2}\bigg)U\bigg(\frac{j+1}{2},\frac{j+3}{2}+\la,\frac{1}{t}\bigg),&j=0,2,4,\ldots.
\end{aligned}
\right.
\eea
Here $U(a,b,z)$ is the Kummer function of the second kind \cite[p. 325]{NIST},
$$
U(a,b,z):=\frac{\Gamma(1-b)}{\Gamma(a-b+1)}{}_{1}F_{1}(a;b;z)+\frac{\Gamma(b-1)}{\Gamma(a)}z^{1-b}{}_{1}F_{1}(a-b+1;2-b;z),\qquad b\notin \mathbb{Z}
$$
and it has the integral representation \cite[p. 326]{NIST}
$$
U(a,b,z)=\frac{1}{\Gamma(a)}\int_{0}^{\infty}\mathrm{e}^{-zs}s^{a-1}(1+s)^{b-a-1}ds,\qquad \Re a>0,\;\Re z>0.
$$

It is well known that our orthogonal polynomials $P_n(x;t)$ have two alternative representations: the first as a determinant and the second as a multiple integral \cite[p. 27]{Szego} (see also \cite[p. 17-18]{Ismail}),
\bea
P_n(x;t)&=&\frac{1}{D_n(t)}\begin{vmatrix}
\mu_{0}(t)&\mu_{1}(t)&\cdots&\mu_{n}(t)\\
\mu_{1}(t)&\mu_{2}(t)&\cdots&\mu_{n+1}(t)\\
\vdots&\vdots&&\vdots\\
\mu_{n-1}(t)&\mu_{n}(t)&\cdots&\mu_{2n-1}(t)\\
1&x&\cdots&x^n
\end{vmatrix}\nonumber\\
&=&\frac{1}{n!\:D_n(t)}\int_{\mathbb{R}^{n}}\prod_{1\leq j<k\leq n}(x_j-x_k)^2\prod_{l=1}^n (x-x_l)w(x_l;t) dx_l.\nonumber
\eea
Furthermore, the Hankel determinant $D_n(t)$ can be expressed as a product \cite[(2.1.6)]{Ismail},
\be\label{hankel}
D_{n}(t)=\prod_{j=0}^{n-1}h_{j}(t).
\ee
It follows from (\ref{be2}) and (\ref{hankel}) that
\be\label{re}
\bt_n(t)=\frac{D_{n+1}(t)D_{n-1}(t)}{D_n^2(t)}.
\ee

One of the main task of this paper is to obtain the large $n$ asymptotic expansions of the recurrence coefficient $\bt_n(t)$, the sub-leading coefficient $\mathrm{p}(n,t)$ and the Hankel determinant $D_n(t)$. For the study of asymptotics of recurrence coefficients of semi-classical orthogonal polynomials and associated Hankel determinants, see \cite{Clarkson,Clarkson1,Kuij,Min2021,Xu2015,Zeng} for reference. See also \cite{BCI,CG,Its,Wu} on the asymptotics for weights with jump discontinuities and Fisher-Hartwig singularities.

This paper is organized as follows. In Section 2, we get two auxiliary quantities $R_n(t)$ and $r_n(t)$ from the ladder operators. Based on identities for $R_n(t)$, $r_n(t)$ and $\bt_n(t)$ obtained from the compatibility conditions ($S_{1}$), ($S_{2}$) and ($S_{2}'$), we derive the second-order difference equation satisfied by $\bt_n(t)$ and it can be transformed into the discrete Painlev\'{e} equation d-$\mathrm{P_{IV}}$. In addition, we obtain the second-order linear differential equation satisfied by the monic orthogonal polynomials.

In Section 3, we study the evolution of the auxiliary quantities in $t$. It is found that $R_n(t)$ and $r_n(t)$ satisfy the coupled Riccati equations, from which we obtain the second-order differential equations for $R_n(t)$ and $r_n(t)$, respectively. We prove that the equation for $R_n(t)$ can be transformed to a particular Painlev\'{e} V after some change of variables. We also obtain the second-order differential equation satisfied by $\bt_n(t)$. As a byproduct, the second-order difference equation for $\mathrm{p}(n,t)$ is derived.

In Section 4, we introduce a quantity $H_n(t)$, related to the logarithmic derivative of the Hankel determinant. We show that $H_n(t)$ can be expressed in terms of $\bt_n(t)$ and $\mathrm{p}(n,t)$, and $H_n(t)$ satisfies a second-order nonlinear differential equation. Based on the relations of $H_n(t)$, $\bt_n(t)$ and $\mathrm{p}(n,t)$, we find that $\mathrm{p}(n,t)$, with a change of variables, satisfies the Jimbo-Miwa-Okamoto $\s$-form of Painlev\'{e} V. In Section 5,
we give an alternative derivation of the Painlev\'{e} V equations for $R_n(t)$ and $\mathrm{p}(n,t)$ by establishing the relation between our problem and \cite{BC,Chen2012}.

In the last section, we study the large $n$ asymptotics of $\bt_n(t)$, $\mathrm{p}(n,t)$, $H_n(t)$ and $D_n(t)$ in the one-cut case. We start from using the Coulomb fluid approach to obtain the large $n$ expansion form of $\bt_n(t)$. By making use of the second-order difference equation satisfied by $\bt_n(t)$, we obtain the full asymptotic expansion for $\bt_n(t)$. The asymptotic expansions for $\mathrm{p}(n,t)$ and $H_n(t)$ are then derived from their relation with $\bt_n(t)$. By the Coulomb fluid approach, we also obtain the large $n$ expansion for the free energy, which provides the expansion form of $\ln D_n(t)$. From the relation of $D_n(t)$ and $\bt_n(t)$ in (\ref{re}), we finally obtain the asymptotic expansion for $\ln D_n(t)$, with two undetermined constants.

\section{Ladder Operators and Its Compatibility Conditions}
The ladder operator approach is a very useful and powerful tool to analyze the recurrence coefficients of orthogonal polynomials and the associated Hankel determinants. See, e.g., \cite{BC,BCE,Dai,Filipuk,Min2020}.
Following the general set-up by Chen and Ismail \cite{ChenIsmail2} (see also Ismail \cite[Chapter 3]{Ismail} and Van Assche \cite[Chapter 4]{VanAssche}), we have the lowering and raising ladder operators for our monic orthogonal polynomials:
\be\label{lowering}
\left(\frac{d}{dz}+B_{n}(z)\right)P_{n}(z)=\beta_{n}A_{n}(z)P_{n-1}(z),
\ee
\be\label{raising}
\left(\frac{d}{dz}-B_{n}(z)-\mathrm{v}'(z)\right)P_{n-1}(z)=-A_{n-1}(z)P_{n}(z),
\ee
where $\mathrm{v}(z):=-\ln w(z)$ is the potential and
\be\label{an}
A_{n}(z):=\frac{1}{h_{n}}\int_{-\infty}^{\infty}\frac{\mathrm{v}'(z)-\mathrm{v}'(y)}{z-y}P_{n}^{2}(y)w(y)dy,
\ee
\be\label{bn}
B_{n}(z):=\frac{1}{h_{n-1}}\int_{-\infty}^{\infty}\frac{\mathrm{v}'(z)-\mathrm{v}'(y)}{z-y}P_{n}(y)P_{n-1}(y)w(y)dy.
\ee
Note that we often suppress the $t$-dependence for brevity, and we have $w(-\infty)=w(+\infty)=0$ in our problem.

The functions $A_n(z)$ and $B_n(z)$ satisfy the following compatibility conditions:
\be
B_{n+1}(z)+B_{n}(z)=z A_{n}(z)-\mathrm{v}'(z), \tag{$S_{1}$}
\ee
\be
1+z(B_{n+1}(z)-B_{n}(z))=\beta_{n+1}A_{n+1}(z)-\beta_{n}A_{n-1}(z), \tag{$S_{2}$}
\ee
\be
B_{n}^{2}(z)+\mathrm{v}'(z)B_{n}(z)+\sum_{j=0}^{n-1}A_{j}(z)=\beta_{n}A_{n}(z)A_{n-1}(z). \tag{$S_{2}'$}
\ee
Here ($S_{2}'$) is obtained from the combination of ($S_{1}$) and ($S_{2}$). The compatibility conditions have the information needed to determine the recurrence coefficient $\bt_n$ and other auxiliary quantities.


For our weight (\ref{weight}), we have
\be\label{pt}
\mathrm{v}(z)=-\ln w(z)=z^2-\la\ln(1+tz^2).
\ee
It follows that
\be\label{vpz}
\mathrm{v}'(z)=2z-\frac{2\la t z}{1+tz^2}
\ee
and
\be\label{vp}
\frac{\mathrm{v}'(z)-\mathrm{v}'(y)}{z-y}=2-\frac{2\la t}{(1+ty^2)(1+tz^2)}+\frac{2\la t^2 yz}{(1+ty^2)(1+tz^2)}.
\ee

Substituting (\ref{vp}) into the definitions of $A_n(z)$ and $B_n(z)$ in (\ref{an}) and (\ref{bn}) and taking account of the parity of the integrands, we obtain
\be\label{anz}
A_{n}(z)=2-\frac{R_{n}(t)}{1+tz^2},
\ee
\be\label{bnz}
B_{n}(z)=\frac{tz\:r_{n}(t)}{1+tz^2},
\ee
where
\be\label{Rnt}
R_{n}(t):=\frac{2\la t}{h_{n}}\int_{-\infty}^{\infty}\frac{P_{n}^{2}(y)w(y)}{1+ty^2}dy,
\ee
\be\label{rnt}
r_{n}(t):=\frac{2\la t}{h_{n-1}}\int_{-\infty}^{\infty}\frac{yP_{n}(y)P_{n-1}(y)w(y)}{1+ty^2} dy.
\ee
\begin{proposition}
The auxiliary quantities $R_n(t),\;r_n(t)$ and the recurrence coefficient $\bt_n$ satisfy the following identities:
\be\label{re1}
R_n(t)=t(2\la-r_{n}(t)-r_{n+1}(t)),
\ee
\be\label{re2}
\bt_n=\frac{n+r_n(t)}{2},
\ee
\be\label{re3}
t\:r_n(t)(2\la-r_n(t))=\bt_n R_n(t)R_{n-1}(t),
\ee
\be\label{re4}
t\: r_n^2(t)-2(1+\la t)r_n(t)-\sum_{j=0}^{n-1}R_j(t)+2\bt_n(R_n(t)+R_{n-1}(t))=0.
\ee
\end{proposition}
\begin{proof}
Substituting (\ref{anz}) and (\ref{bnz}) into ($S_{1}$), we get (\ref{re1}). Similarly, substituting (\ref{anz}) and (\ref{bnz}) into ($S_{2}'$), we obtain (\ref{re2}), (\ref{re3}) and (\ref{re4}).
\end{proof}
\begin{theorem}\label{the}
The recurrence coefficient $\bt_n$ satisfies the second-order nonlinear difference equation
\be\label{beta}
(n-2\bt_n)(n+2\la-2\bt_n)+t\bt_n(2n+1+2\la-2\bt_n-2\bt_{n+1})(2n-1+2\la-2\bt_{n-1}-2\bt_{n})=0,
\ee
with the initial conditions
$$
\bt_0=0,\qquad \bt_1=\frac{U\big(\frac{3}{2},\frac{5}{2}+\la,\frac{1}{t}\big)}{2t\:U\big(\frac{1}{2},\frac{3}{2}+\la,\frac{1}{t}\big)}.
$$
\end{theorem}
\begin{proof}
From (\ref{re2}) we have
\be\label{rne}
r_n(t)=2\bt_n-n.
\ee
Substituting it into (\ref{re1}) gives
\be\label{Rne}
R_n(t)=t(2n+1+2\la-2\bt_n-2\bt_{n+1}).
\ee
Using (\ref{rne}) and (\ref{Rne}) to eliminate $r_n(t), R_n(t)$ and $R_{n-1}(t)$ in (\ref{re3}), we obtain (\ref{beta}). The initial condition $\bt_0=0$ is shown in the introduction and $\bt_1$ is evaluated by $\bt_1=h_1/h_0=\mu_2(t)/\mu_0(t)$ with the aid of (\ref{moment}).
\end{proof}
\begin{remark}
Letting $x_n=2\bt_n-n-\la$, then equation (\ref{beta}) becomes
$$
(x_n+x_{n-1})(x_n+x_{n+1})=-\frac{2(x_n^2-\la^2)}{t(x_n+n+\la)},
$$
which can be identified with the discrete Painlev\'{e} equation d-$P_{IV}$ \cite[p. 269]{GR} (see also \cite[(1.26)]{VanAssche})
$$
(x_n+x_{n-1})(x_n+x_{n+1})=\frac{(x_n^2-a^2)(x_n^2-b^2)}{(x_n+z_n)^2-c^2},
$$
where $z_n=n+\frac{\la}{2}+\frac{b^2 t}{4},\;a=\la, \;c=\frac{b^2 t}{4}-\frac{\la}{2}$ and $b\rightarrow\infty$. In the geometric aspects of Painlev\'{e} equations, this d-$P_{IV}$ corresponds to d-$P(E_{6}^{(1)}/A_{2}^{(1)})$ with symmetry $E_{6}^{(1)}$ and surface $A_{2}^{(1)}$ in the Okamoto-Sakai classification of Kajiwara, Noumi and Yamada \cite[Section 8.1.15]{KNY} after identifying $f=x_{2n-1}$ and $g=x_{2n}$.
\end{remark}

In the end of this section, we show that our monic orthogonal polynomials satisfy a second-order differential equation, with the coefficients all expressed in terms of $\bt_n$ and $\bt_{n+1}$.
\begin{theorem}
The monic orthogonal polynomials $P_n(z),\; n=0,1,2,\ldots,$ satisfy the following second-order differential equation:
\be\label{ode}
P_n''(z)-\left(\mathrm{v}'(z)+\frac{A_{n}'(z)}{A_{n}(z)}\right)P_n'(z)+\left(B_{n}'(z)-B_{n}(z)\frac{A_{n}'(z)}{A_{n}(z)}
+\sum_{j=0}^{n-1}A_{j}(z)\right)P_n(z)=0,
\ee
where $\mathrm{v}'(z)$ is given by (\ref{vpz}) and
\be\label{anz1}
A_n(z)=2-\frac{t(2n+1+2\la-2\bt_n-2\bt_{n+1})}{1+tz^2},
\ee
\be\label{bnz1}
B_n(z)=\frac{tz(2\bt_n-n)}{1+tz^2},
\ee
\bea\label{sum1}
\sum_{j=0}^{n-1}A_{j}(z)&=&2n+\frac{2(n-2\bt_n)(n+2\la-2\bt_n)}{(2n+1+2\la-2\bt_n-2\bt_{n+1})(1+tz^2)}\nonumber\\
&-&\frac{n^2t+2n(1+\la t)+2(t-2)\bt_n-4t\bt_n\bt_{n+1}}{1+tz^2}.
\eea
\end{theorem}
\begin{proof}
Eliminating $P_{n-1}(z)$ from the ladder operator equations (\ref{lowering}) and (\ref{raising}), we obtain (\ref{ode}) with the aid of ($S_{2}'$). Using (\ref{anz}) and (\ref{Rne}), we arrive at (\ref{anz1}). The expression of $B_n(z)$ in (\ref{bnz1}) comes from (\ref{bnz}) and (\ref{rne}).

From (\ref{anz}) we have
\bea
\sum_{j=0}^{n-1}A_{j}(z)&=&2n-\frac{\sum_{j=0}^{n-1}R_{j}(t)}{1+tz^2}\nonumber\\
&=&2n-\frac{t\: r_n^2(t)-2(1+\la t)r_n(t)+2\bt_nR_n(t)+\frac{2t\:r_n(t)(2\la-r_n(t))}{R_n(t)}}{1+tz^2},\nonumber
\eea
where use has been made of (\ref{re4}) and (\ref{re3}). Inserting (\ref{rne}) and (\ref{Rne}) into the above, we obtain (\ref{sum1}). This completes the proof.
\end{proof}
\begin{remark}
When $\la=0$, equation (\ref{ode}) is reduced to Hermite's differential equation \cite[p. 106]{Szego}
$$
P_n''(z)-2zP_n'(z)+2nP_n(z)=0.
$$
\end{remark}

\section{$t$ Evolution and Painlev\'{e} V}
Since all the quantities considered in this paper, such as the recurrence coefficient $\bt_n(t)$, the sub-leading coefficient $\mathrm{p}(n,t)$ and the auxiliary quantities $R_n(t)$ and $r_n(t)$, depend on $t$, we would like to study the evolution of these quantities in $t$ in this section. We start from taking derivatives with respect to $t$ in the orthogonality conditions to obtain the following proposition.

\begin{proposition}
We have
\be\label{p1}
2t^2\frac{d}{dt}\ln h_n(t)=2\la t-R_n(t),
\ee
\be\label{p2}
2t^2\frac{d}{dt}\mathrm{p}(n,t)=r_n(t)-\bt_n R_n(t).
\ee
\end{proposition}
\begin{proof}
Since
$$
h_n(t)=\int_{-\infty}^{\infty}P_n^2(x;t)\mathrm{e}^{-x^2}(1+tx^2)^\la dx,
$$
we find
$$
h_n'(t)=\la\int_{-\infty}^{\infty}x^2P_n^2(x;t)\mathrm{e}^{-x^2}(1+tx^2)^{\la-1} dx.
$$
It follows that
\bea
th_n'(t)&=&\la\int_{-\infty}^{\infty}(1+tx^2-1)P_n^2(x;t)\mathrm{e}^{-x^2}(1+tx^2)^{\la-1} dx\nonumber\\
&=&\la h_n-\la\int_{-\infty}^{\infty}\frac{P_n^2(x;t)w(x;t)}{1+tx^2} dx.\nonumber
\eea
In view of (\ref{Rnt}), we arrive at (\ref{p1}).

Next, we get from (\ref{or}) with $m$ replaced by $n-2$ that
$$
0=\int_{-\infty}^{\infty}P_n(x;t)P_{n-2}(x;t)\mathrm{e}^{-x^2}(1+tx^2)^\la dx.
$$
Taking a derivative with respect to $t$ gives
\be\label{dp}
\frac{d}{dt}\mathrm{p}(n,t)=-\frac{\la}{h_{n-2}}\int_{-\infty}^{\infty}x^2 P_n(x;t)P_{n-2}(x;t)\mathrm{e}^{-x^2}(1+tx^2)^{\la-1} dx.
\ee
Then we have
\bea
t\frac{d}{dt}\mathrm{p}(n,t)&=&-\frac{\la}{h_{n-2}}\int_{-\infty}^{\infty}(1+tx^2-1) P_n(x;t)P_{n-2}(x;t)\mathrm{e}^{-x^2}(1+tx^2)^{\la-1} dx\nonumber\\
&=&\frac{\la}{h_{n-2}}\int_{-\infty}^{\infty}P_n(x;t)P_{n-2}(x;t)\mathrm{e}^{-x^2}(1+tx^2)^{\la-1} dx.\nonumber
\eea
With the aid of the recurrence relation (\ref{rr}) and (\ref{be2}), we find that
$$
t\frac{d}{dt}\mathrm{p}(n,t)=\frac{\la}{h_{n-1}}\int_{-\infty}^{\infty}\frac{xP_n(x;t)P_{n-1}(x;t)w(x;t)}{1+tx^2} dx
-\frac{\la}{h_{n-1}}\int_{-\infty}^{\infty}\frac{P_n^2(x;t)w(x;t)}{1+tx^2} dx.
$$
By using (\ref{Rnt}) and (\ref{rnt}), we obtain (\ref{p2}).
\end{proof}

\begin{corollary}
The following identity holds:
\be\label{iden}
\bt_nR_n(t)=r_n(t)+2t\: \mathrm{p}(n,t)+2t\bt_n\bt_{n-1}.
\ee
\end{corollary}
\begin{proof}
Through integration by parts, we find from (\ref{dp}) that
$$
2t\frac{d}{dt}\mathrm{p}(n,t)=\frac{1}{h_{n-2}}\int_{-\infty}^{\infty}\big(xP_n'(x)P_{n-2}(x)-2x^2P_n(x)P_{n-2}(x)\big)w(x)dx.
$$
Using the fact that
$$
xP_n'(x)=nP_n(x)-2\mathrm{p}(n,t)P_{n-2}(x)+\mathrm{lower}\:\: \mathrm{degree}\:\:\mathrm{polynomials},
$$
we have
\be\label{dp1}
2t\frac{d}{dt}\mathrm{p}(n,t)=-2\mathrm{p}(n,t)-2\bt_n\bt_{n-1}.
\ee
The combination of (\ref{p2}) and (\ref{dp1}) gives the result in (\ref{iden}).
\end{proof}
\begin{theorem}
The sub-leading coefficient of the monic orthogonal polynomials, $\mathrm{p}(n):=\mathrm{p}(n,t)$, satisfies the second-order nonlinear difference equation
\bea\label{pnd}
&&\big(n-2\mathrm{p}(n)+2\mathrm{p}(n+1)\big)\big(n+2\la-2\mathrm{p}(n)+2\mathrm{p}(n+1)\big)-\big(2n-1+2\la-2\mathrm{p}(n-1)+2\mathrm{p}(n+1)\big)\nonumber\\
&\times&\left[n-2(t+1)\mathrm{p}(n)+2\mathrm{p}(n+1)-2t\big(\mathrm{p}(n-1)-\mathrm{p}(n)\big)\big(\mathrm{p}(n)-\mathrm{p}(n+1)\big)\right]=0,
\eea
with the initial conditions
$$
\mathrm{p}(1,t)=0,\qquad \mathrm{p}(2,t)=-\frac{U\big(\frac{3}{2},\frac{5}{2}+\la,\frac{1}{t}\big)}{2t\:U\big(\frac{1}{2},\frac{3}{2}+\la,\frac{1}{t}\big)}.
$$
\end{theorem}
\begin{proof}
Multiplying by $R_{n-1}(t)$ on both sides of (\ref{iden}) and comparing with (\ref{re3}), we have
$$
t\:r_n(t)(2\la-r_n(t))=\big(r_n(t)+2t\: \mathrm{p}(n,t)+2t\bt_n\bt_{n-1}\big)R_{n-1}(t).
$$
Eliminating $r_n(t)$ and $R_{n-1}(t)$ by using (\ref{rne}) and (\ref{Rne}), we get an equation satisfied by $\bt_n,\;\bt_{n-1}$ and $\mathrm{p}(n,t)$. Taking account of the relation (\ref{be1}), we obtain the difference equation (\ref{pnd}). The initial condition $\mathrm{p}(1,t)=0$ is known and $\mathrm{p}(2,t)$ is evaluated from the fact that $\mathrm{p}(2,t)=\mathrm{p}(1,t)-\bt_1=-\bt_1$.
\end{proof}
\begin{remark}
Substituting (\ref{be1}) into (\ref{beta}), one would only obtain a third-order difference equation satisfied by $\mathrm{p}(n,t)$.
\end{remark}
\begin{proposition}
The auxiliary quantities $R_n(t)$ and $r_n(t)$ satisfy the coupled Riccati equations:
\be\label{ri1}
t^2 r_n'(t)=\frac{t\: r_n(t)(2\la-r_n(t))}{R_n(t)}-\frac{n+r_n(t)}{2}R_n(t),
\ee
\be\label{ri2}
2 t^2 R_n'(t)=R_n^2(t)+(2 t\: r_n(t)-2 \la  t+t-2)R_n(t)-4 t\: r_n(t)+4 \la  t.
\ee
\end{proposition}

\begin{proof}
From (\ref{be2}) we have
$$
\ln\bt_n=\ln h_n-\ln h_{n-1}.
$$
Using (\ref{p1}), it follows that
$$
2t^2\frac{d}{dt}\ln\bt_n=R_{n-1}(t)-R_n(t).
$$
That is,
\be\label{btd1}
2t^2\bt_n'(t)=\bt_nR_{n-1}(t)-\bt_nR_n(t).
\ee
By making use of (\ref{re2}) and (\ref{re3}), we obtain (\ref{ri1}).

Next, taking a derivative with respect to $t$ in (\ref{iden}) gives us
\be\label{equa}
(n+r_n(t))R_n'(t)+(R_n(t)-2)r_n'(t)-4t(\bt_n\bt_{n-1})'=0,
\ee
where use has been made of (\ref{re2}) and (\ref{dp1}). To derive the expression of $\bt_n\bt_{n-1}$, we replace $n$ by $n-1$ in (\ref{re1}) and multiply by $\bt_n$ on both sides to obtain
\be\label{bb}
4\bt_n\bt_{n-1}=(n+r_n(t))(n-1+2\la-r_n(t))-\frac{2r_n(t)(2\la-r_n(t))}{R_n(t)},
\ee
where we have used (\ref{re2}) and (\ref{re3}). Substituting (\ref{bb}) into (\ref{equa}), and using (\ref{ri1}) to eliminate $r_n'(t)$, we obtain
\bea
&&\left[2 t^2 R_n'(t)-R_n^2(t)-(2 t\: r_n(t)-2 \la  t+t-2)R_n(t)+4 t\: r_n(t)-4 \la  t\right]\nonumber\\
&\times&\left[(n+r_n(t))R_n^2(t) +2 t\: r_n(t) (r_n(t)-2 \la )\right]=0.\no
\eea
Obviously, this leads to two equations. However, the latter algebraic equation does not hold and should be discarded. (This is because the combination of this algebraic equation and (\ref{re1}) will produce a first-order difference equation for $r_n(t)$ and then for $\bt_n$, which makes a contradiction with Theorem \ref{the}.) Hence, we arrive at (\ref{ri2}).
\end{proof}
The following theorem reveals the relation between our problem and Painlev\'{e} equations.
\begin{theorem}\label{pve}
The auxiliary quantity $R_n(t)$, related to $\bt_n$ in (\ref{Rne}), satisfies the second-order nonlinear ordinary differential equation
\begin{small}
\bea\label{Rnd}
&&4 t^4 R_n(R_n-2)  R_n''-4 t^4 (R_n-1) (R_n')^2+4 t^3  R_n(R_n-2) R_n'-R_n^5+\left[(2 n+1+2 \la)t +5\right]R_n^4\nonumber\\
&-&4\left[(2 n+1+2 \la)t+2\right]R_n^3+\left[(1-4 \la ^2) t^2+4 (2 n+1+2\la)t+4\right]R_n^2+16 \la ^2 t^2 (R_n-1)=0.
\eea
\end{small}
Let $t=\frac{1}{s}$ and $R_n(t)=\frac{2W_n(s)}{W_n(s)-1}$. Then $W_n(s)$ satisfies the Painlev\'{e} V equation \cite{Gromak}
\be\label{pv}
W_{n}''=\frac{(3 W_n-1) (W_n')^2}{2W_n (W_n-1) }-\frac{W_n'}{s}+\frac{(W_n-1)^2 }{s^2}\left(\mu_1 W_n +\frac{\mu_2}{W_n}\right)+\frac{\mu_3 W_n}{s}+\frac{\mu_4 W_n(W_n+1) }{W_n-1},
\ee
with the parameters
$$
\mu_{1}=\frac{1}{8},\qquad \mu_2=-\frac{\la ^2}{2},\qquad \mu_3=n+\la+\frac{1}{2},\qquad \mu_4=-\frac{1}{2}.
$$
\end{theorem}
\begin{proof}
Solving for $r_n(t)$ from (\ref{ri2}) and substituting it into (\ref{ri1}), we arrive at (\ref{Rnd}). By the change of variables, equation (\ref{Rnd}) is transformed into (\ref{pv}).
\end{proof}
\begin{theorem}
The quantity $r_n(t)$ satisfies the second-order nonlinear differential equation
\begin{small}
\bea\label{rnd}
&&\big\{2 t^5 r_n' r_n''+ t^3(2 \la t +3 t+2-2 t\: r_n)(r_n')^2-2t^2\left[3r_n^2+2 (n-2 \la ) r_n-2   n\la\right]r_n'\nonumber\\
&+&4 r_n(r_n+n)(r_n-2\la)(t\:r_n-\la t-1)\big\}^2=t\left[t^3 (r_n')^2-2 r_n(r_n+n)(r_n-2\la)\right]\nonumber\\
&\times&\left[2 t^3 r_n''+ t(2 \la  t+3 t+2-2 t\: r_n)r_n'-6 r_n^2-4 (n-2 \la ) r_n+4   n\la\right]^2.
\eea
\end{small}
Then, the recurrence coefficient $\bt_n$ satisfies the second-order differential equation
\begin{small}
\bea\label{btd}
&&\big\{2 t^5 \bt_n' \bt_n''- t^3\left[(4\bt_n-2n-3-2 \la) t-2\right](\bt_n')^2-t^2\left[12\bt_n^2-8(n+\la)\bt_n+n(n+2\la)\right]\bt_n'\nonumber\\
&+& 2\bt_n(2\bt_n-n)(2\bt_n-n-2\la)\left[(2\bt_n-n-\la)t-1\right]\big\}^2=t\left[t^3 (\bt_n')^2-\bt_n (2\bt_n-n)(2\bt_n-n-2\la)\right]\nonumber\\
&\times&\big\{2 t^3 \bt_n''- t\left[(4\bt_n-2n-3-2 \la) t-2\right]\bt_n'-12\bt_n^2+8(n+\la)\bt_n-n(n+2\la)\big\}^2.
\eea
\end{small}
\end{theorem}

\begin{proof}
Solving for $R_n(t)$ from (\ref{ri1}) and substituting either solution into (\ref{ri2}), we obtain (\ref{rnd}) after removing the square roots. By making use of (\ref{re2}), equation (\ref{rnd}) is converted into (\ref{btd}).
\end{proof}

\section{Logarithmic Derivative of the Hankel Determinant, Sub-leading Coefficient and $\s$-form of Painlev\'{e} V}
We start from defining a quantity
\be\label{def}
H_n(t):=-\sum_{j=0}^{n-1}R_j(t).
\ee
It is easy to see from (\ref{hankel}) and (\ref{p1}) that $H_n(t)$ is related to the logarithmic derivative of the Hankel determinant as follows:
\be\label{hd}
H_n(t)=2t^2\frac{d}{dt}\ln D_n(t)-2n\la t.
\ee
\begin{proposition}
The quantity $H_n(t)$ can be expressed in terms of $\bt_n$ and $\mathrm{p}(n,t)$ by
\be\label{p3}
H_n(t)=2t\bt_n-4t\:\mathrm{p}(n,t)-nt(n+2\la).
\ee
\end{proposition}
\begin{proof}
Using (\ref{re1}) and (\ref{re2}), we find from (\ref{def}) that
\bea
H_n(t)&=&2t\sum_{j=0}^{n-1}\bt_j+2t\sum_{j=0}^{n-1}\bt_{j+1}-nt(n+2\la)\nonumber\\
&=&4t\sum_{j=0}^{n-1}\bt_j+2t\bt_n-nt(n+2\la).\nonumber
\eea
In view of (\ref{sum}), we obtain the desired result.
\end{proof}
\begin{theorem}
The quantity $H_n(t)$ satisfies the following second-order differential equation
\begin{small}
\bea\label{hnd}
&&\big\{4 t^7 (H_n'')^2-4t^4( t^2 H_n'-t H_n+2 \la  t-2)H_n''-t^3 \left[20 t H_n+(4 \la ^2-1) t^2+8(4n-\la)t+4\right](H_n')^2\nonumber\\
&+&8 t^5 (H_n')^3+2 t^2 \left[8 t H_n^2+ \big((4 \la ^2-1) t^2+16 (n-\la) t+12\big)H_n+2 (\la  t-3) \big(4 n (\la  t-1)+t\big)\right]H_n'\nonumber\\
&-&4 t^2 H_n^3+t  \left[(1-4 \la ^2) t^2-8(n-2\la)t-12\right]H_n^2-4\big[\la t^3 (2n\la-2 \la ^2+1)-2t^2 (4n\la-3 \la^2+1)\nonumber\\
&+&6 t (n-\la )+2\big]H_n+8 (\la  t-1)^2  \left[2 n (\la  t-1)+t\right]\big\}^2=16 \left[t H_n-2 t^2 H_n'+(\la  t-1)^2\right]\big\{2 t^4 H_n''-2 t H_n^2\nonumber\\
&+&t^2(4 H_n-t+8 n) H_n'+\left[(1-2 \la ^2) t^2-4 (n-\la)t-2\right]H_n-2 (\la  t-1) \left[2 n (\la  t-1)+t\right]\big\}^2.
\eea
\end{small}
\end{theorem}
\begin{proof}
Let us denote by
$$
X:=\bt_n R_{n-1}(t),
$$
$$
Y:=\bt_n R_n(t).
$$
Then equations (\ref{re4}) and (\ref{btd1}) become
$$
t\: r_n^2(t)-2(1+\la t)r_n(t)+H_n(t)+2X+2Y=0
$$
and
$$
X-Y-t^2 r_n'(t)=0,
$$
where use has been made of (\ref{re2}).
Solving this linear system for $X$ and $Y$, we have
\be\label{X}
X=-\frac{1}{4}\left[H_n(t)-2t^2r_n'(t)+t\:r_n^2(t)-2(\la t+1)r_n(t)\right],
\ee
\be\label{Y}
Y=-\frac{1}{4}\left[H_n(t)+2t^2r_n'(t)+t\:r_n^2(t)-2(\la t+1)r_n(t)\right].
\ee
From (\ref{re3}) and with the aid of (\ref{re2}), we find
\be\label{pro}
X\cdot Y=\frac{t\:r_n(t)(2\la-r_n(t))(n+r_n(t))}{2}.
\ee
Substituting (\ref{X}) and (\ref{Y}) into (\ref{pro}) gives
\bea\label{hr}
&&\left[H_n(t)-2t^2r_n'(t)+t\:r_n^2(t)-2(\la t+1)r_n(t)\right]\left[H_n(t)+2t^2r_n'(t)+t\:r_n^2(t)-2(\la t+1)r_n(t)\right]\nonumber\\
&=&8t\:r_n(t)(2\la-r_n(t))(n+r_n(t)).
\eea

On the other hand, the combination of (\ref{p3}) and (\ref{re2}) shows that $\mathrm{p}(n,t)$ can be expressed in terms of $H_n(t)$ and $r_n(t)$. Substituting this expression for $\mathrm{p}(n,t)$ and (\ref{Y}) for $\bt_nR_n(t)$ into (\ref{p2}), we obtain a quadratic equation for $r_n(t)$,
\be\label{rh}
t\:r_n^2(t)+2(1-\la t)r_n(t)+2 t H_n'(t)-H_n(t)=0.
\ee
Substituting either solution for $r_n(t)$ into (\ref{hr}), we arrive at equation (\ref{hnd}) after removing the square roots.
\end{proof}
\begin{remark}
From (\ref{def}) we have
$$
R_n(t)=H_n(t)-H_{n+1}(t).
$$
Then one can also derive the second-order difference equation satisfied by $H_n(t)$ from (\ref{re2}), (\ref{re3}) and (\ref{re4}). We would not present the result here and leave it to the reader.
\end{remark}
\begin{theorem}\label{pnt}
The sub-leading coefficient $\mathrm{p}(n,t)$ satisfies the second-order nonlinear differential equation
\bea\label{sod}
&&16 t^6 (\mathrm{p}''(n,t))^2+64 t^5 \mathrm{p}'(n,t) \mathrm{p}''(n,t)-64 t^5 (\mathrm{p}'(n,t))^3\nonumber\\
&-&4t^2(\mathrm{p}'(n,t))^2\left[16 t^2 \mathrm{p}(n,t)+(2 n+3+2 \la)(2 n-5+2 \la) t^2+4t (2 n-1-2 \la)+4\right]\nonumber\\
&-&4t\:\mathrm{p}'(n,t)\left[4 \big((2 n-1-2\la) t+2\big) \mathrm{p}(n,t)+  n(n-1)\big((2 n-1+2\la) t+2\big) \right]\nonumber\\
&-&16 \mathrm{p}^2(n,t)-8 n(n-1)  \mathrm{p}(n,t)-n^2(n-1)^2 =0,
\eea
where $\mathrm{p}'(n,t):=\frac{d}{dt}\mathrm{p}(n,t),\;\mathrm{p}''(n,t):=\frac{d^{2}}{dt^{2}}\mathrm{p}(n,t)$.
Let $t=\frac{1}{s}$ and $\mathrm{p}(n,t)=\s_n(s)-\frac{n(n-1)}{4}$. Then $\s_n(s)$ satisfies the Jimbo-Miwa-Okamoto $\s$-form of Painlev\'{e} V \cite[(C.45)]{Jimbo1981},
\be\label{jmo}
\big(s\s_{n}''\big)^2=\big[\s_{n}-s\s_{n}'+2\big(\s_{n}'\big)^2+(\nu_{0}+\nu_{1}+\nu_{2}+\nu_{3})\s_{n}'\big]^2-4\big(\nu_{0}+\s_{n}'\big)
\big(\nu_{1}+\s_{n}'\big)\big(\nu_{2}+\s_{n}'\big)\big(\nu_{3}+\s_{n}'\big),
\ee
with parameters $\nu_{0}=0,\; \nu_{1}=-\frac{n}{2},\; \nu_{2}=\la,\; \nu_{3}=-\frac{n-1}{2}$.
\end{theorem}
\begin{proof}
Substituting the expression of $\bt_nR_n(t)$ in (\ref{Y}) into (\ref{p2}), we have
$$
2t^2 r_n'(t)=8t^2\mathrm{p}'(n,t)+2(\la t-1)r_n(t)-t\: r_n^2(t)-H_n(t).
$$
Plugging it into (\ref{hr}) produces a linear equation for $H_n(t)$, and we find
\be\label{hn}
H_n(t)=\frac{8t^4(\mathrm{p}'(n,t))^2-t\:r_n^2(t)(n+2t^2\mathrm{p}'(n,t))+2t\:r_n(t)\left[n\la+2t(\la t-1)\mathrm{p}'(n,t)\right]}{2t^2\mathrm{p}'(n,t)-r_n(t)}.
\ee
With the aid of (\ref{re2}), we write (\ref{p3}) as
\be\label{hrp}
H_n(t)=t(n+r_n(t))-4t\:\mathrm{p}(n,t)-nt(n+2\la).
\ee
Subtracting (\ref{hrp}) from (\ref{hn}) gives us a quadratic equation for $r_n(t)$:
\bea
&&r_n^2(t)\left(n-1+2t^2\mathrm{p}'(n,t)\right)+r_n(t)\left[n(n-1)+4\mathrm{p}(n,t)+2t(t+2-2\la t)\mathrm{p}'(n,t)\right]\nonumber\\
&-&2t^2 \mathrm{p}'(n,t)\left[n(n-1+2\la)+4\mathrm{p}(n,t)+4t\:\mathrm{p}'(n,t)\right]=0.\nonumber
\eea
Then $r_n(t)$ can be expressed in terms of $\mathrm{p}(n,t)$ and $\mathrm{p}'(n,t)$, denoted by
\be\label{rp}
r_n(t)=f\left[\mathrm{p}(n,t),\mathrm{p}'(n,t)\right],
\ee
where $f\left[\mathrm{p}(n,t),\mathrm{p}'(n,t)\right]$ is an expression of $\mathrm{p}(n,t)$ and $\mathrm{p}'(n,t)$ and is explicitly known.
It follows from (\ref{hrp}) that
\be\label{hp}
H_n(t)=t\left(n+f\left[\mathrm{p}(n,t),\mathrm{p}'(n,t)\right]\right)-4t\:\mathrm{p}(n,t)-nt(n+2\la).
\ee
Substituting (\ref{rp}) and (\ref{hp}) into (\ref{rh}), we obtain the second-order differential equation satisfied by $\mathrm{p}(n,t)$ in (\ref{sod}) after clearing the square roots. Under the given change of variables, equation (\ref{sod}) is transformed into (\ref{jmo}).
\end{proof}

\section{An Alternative Derivation of the Painlev\'{e} V Equation}
In the papers \cite{BC} and \cite{Chen2012}, the authors studied the monic polynomials orthogonal with respect to a deformed Laguerre weight
$$
\tilde{w}(x;t,\al)=x^{\al}\mathrm{e}^{-x}(x+t)^{\lambda},\qquad x\in \mathbb{R}^{+},\; \al>-1,\; t>0
$$
and the orthogonality is
$$
\int_{0}^{\infty}\tilde{P}_m(x;t,\al)\tilde{P}_n(x;t,\al)x^{\al}\mathrm{e}^{-x}(x+t)^{\lambda}dx=\tilde{h}_n(t,\al)\delta_{mn},
$$
where $\tilde{P}_n(x;t,\al), n=0,1,2,\ldots$ are monic polynomials of degree $n$ and have the expansion
$$
\tilde{P}_n(x;t,\al)=x^n+\mathrm{\tilde{p}}(n,t,\al)x^{n-1}+\cdots+\tilde{P}_n(0;t,\al).
$$
The associated Hankel determinant is
$$
\tilde{D}_n(t,\al):=\det\left(\int_{0}^{\infty}x^{j+k}x^{\al}\mathrm{e}^{-x}(x+t)^{\lambda}dx\right)_{j,k=0}^{n-1}.
$$

By introducing an auxiliary variable
$$
\tilde{R}_n(t,\al)=\frac{\lambda}{\tilde{h}_n(t,\al)}\int_{0}^{\infty}\frac{\tilde{P}_n^2(x;t,\al)}{x+t}\tilde{w}(x;t,\al)dx,
$$
Basor and Chen \cite{BC} proved that
$$
y(t,\al)=\frac{\tilde{R}_n(t,\al)}{\tilde{R}_n(t,\al)-1}
$$
satisfies the Painlev\'{e} V equation
$$
y''=\frac{(3 y-1) (y')^2}{2y (y-1) }-\frac{y'}{t}+\frac{(y-1)^2 }{t^2}\left(\mu_1 y +\frac{\mu_2}{y}\right)+\frac{\mu_3 y}{t}+\frac{\mu_4 y(y+1) }{y-1}
$$
with parameters
$$
\mu_{1}=\frac{\al^2}{2},\qquad \mu_2=-\frac{\la^2}{2},\qquad \mu_3=2n+1+\al+\la,\qquad \mu_4=-\frac{1}{2},
$$
and denoted as
$$
P_{V}\left(\frac{\al^2}{2},-\frac{\la^2}{2},2n+1+\al+\la,-\frac{1}{2}\right).
$$
Chen and McKay \cite{Chen2012} (see also \cite{BC}) found that $\tilde{\s}_n(t,\al)=n(n+\al)+\mathrm{\tilde{p}}(n,t,\al)$ satisfies the $\s$-form of Painlev\'{e} V
\be\label{jmo1}
\big(s\tilde{\s}_{n}''\big)^2=\big[\tilde{\s}_{n}-t\tilde{\s}_{n}'+2\big(\tilde{\s}_{n}'\big)^2+(\nu_{0}+\nu_{1}+\nu_{2}+\nu_{3})\tilde{\s}_{n}'\big]^2
-4\big(\nu_{0}+\tilde{\s}_{n}'\big)
\big(\nu_{1}+\tilde{\s}_{n}'\big)\big(\nu_{2}+\tilde{\s}_{n}'\big)\big(\nu_{3}+\tilde{\s}_{n}'\big)
\ee
with parameters $\nu_{0}=0,\; \nu_{1}=-n,\; \nu_{2}=\la,\; \nu_{3}=-n-\al$.

Following the similar procedure in \cite{Min2020}, we will establish the relation between our problem and \cite{BC,Chen2012}.
From the orthogonality (\ref{or}) we have for $m, n=0, 1, 2, \ldots$,
\bea
h_{2n}(t)\delta_{2m,2n}&=&\int_{-\infty}^{\infty}P_{2m}(x)P_{2n}(x)\mathrm{e}^{-x^2}\left(1+t\:x^2\right)^\la dx\no\\
&=&2\int_{0}^{\infty}P_{2m}(x)P_{2n}(x)\mathrm{e}^{-x^2}\left(1+t\:x^2\right)^\la dx\no\\
&=&t^{\la}\int_{0}^{\infty}P_{2m}(\sqrt{x})P_{2n}(\sqrt{x})x^{-\frac{1}{2}}\mathrm{e}^{-x}\left(x+\frac{1}{t}\right)^{\la} dx\no
\eea
and
\bea
h_{2n+1}(t)\delta_{2m+1,2n+1}&=&\int_{-\infty}^{\infty}P_{2m+1}(x)P_{2n+1}(x)\mathrm{e}^{-x^2}\left(1+t\:x^2\right)^\la dx\no\\
&=&2\int_{0}^{\infty}P_{2m+1}(x)P_{2n+1}(x)\mathrm{e}^{-x^2}\left(1+t\:x^2\right)^\la dx\no\\
&=&t^{\la}\int_{0}^{\infty}\frac{P_{2m+1}(\sqrt{x})}{\sqrt{x}}\frac{P_{2n+1}(\sqrt{x})}{\sqrt{x}}x^{\frac{1}{2}}\mathrm{e}^{-x}\left(x+\frac{1}{t}\right)^{\la} dx.\no
\eea
It follows that for $n=0, 1, 2,\ldots$,
\be\label{rela}
\tilde{P}_{n}\left(x;\frac{1}{t},-\frac{1}{2}\right)=P_{2n}(\sqrt{x}),\qquad\qquad \tilde{P}_{n}\left(x;\frac{1}{t},\frac{1}{2}\right)=\frac{P_{2n+1}(\sqrt{x})}{\sqrt{x}},
\ee
and
\be\label{rela1}
\tilde{h}_n\left(\frac{1}{t},-\frac{1}{2}\right)=\frac{h_{2n}(t)}{t^{\la}},\qquad\qquad \tilde{h}_n\left(\frac{1}{t},\frac{1}{2}\right)=\frac{h_{2n+1}(t)}{t^{\la}}.
\ee
According to the definition of $R_n(t)$ in (\ref{Rnt}) and using the above results, we find
$$
R_{2n}(t)=2\tilde{R}_{n}\left(\frac{1}{t},-\frac{1}{2}\right),\qquad\qquad R_{2n+1}(t)=2\tilde{R}_{n}\left(\frac{1}{t},\frac{1}{2}\right),\qquad n=0, 1, 2,\ldots.
$$
Since
$$
y\left(t,-\frac{1}{2}\right)=\frac{\tilde{R}_n\left(t,-\frac{1}{2}\right)}{\tilde{R}_n\left(t,-\frac{1}{2}\right)-1}
$$
and
$$
y\left(t,\frac{1}{2}\right)=\frac{\tilde{R}_n\left(t,\frac{1}{2}\right)}{\tilde{R}_n\left(t,\frac{1}{2}\right)-1}
$$
satisfies the Painlev\'{e} V equations
$$
P_{V}\left(\frac{1}{8}, -\frac{\la^2}{2}, 2n+\la+\frac{1}{2}, -\frac{1}{2}\right)
$$
and
$$
P_{V}\left(\frac{1}{8}, -\frac{\la^2}{2}, 2n+\la+\frac{3}{2}, -\frac{1}{2}\right),
$$
respectively, we finally obtain the result in Theorem \ref{pve}.

To proceed, from (\ref{rela}) we have for $n=0, 1, 2,\ldots$,
$$
\mathrm{\tilde{p}}\left(n,\frac{1}{t},-\frac{1}{2}\right)=\mathrm{p}(2n,t),\qquad\qquad \mathrm{\tilde{p}}\left(n,\frac{1}{t},\frac{1}{2}\right)=\mathrm{p}(2n+1,t).
$$
Note that
$$
\tilde{\s}_n\left(t,-\frac{1}{2}\right)=\mathrm{\tilde{p}}\left(n,t,-\frac{1}{2}\right)+n\left(n-\frac{1}{2}\right)
$$
and
$$
\tilde{\s}_n\left(t,\frac{1}{2}\right)=\mathrm{\tilde{p}}\left(n,t,\frac{1}{2}\right)+n\left(n+\frac{1}{2}\right)
$$
satisfy the $\s$-form of Painlev\'{e} V in (\ref{jmo1}) with parameters
$$
\nu_{0}=0,\; \nu_{1}=-n,\; \nu_{2}=\la,\; \nu_{3}=-n+\frac{1}{2}
$$
and
$$
\nu_{0}=0,\; \nu_{1}=-n,\; \nu_{2}=\la,\; \nu_{3}=-n-\frac{1}{2},
$$
respectively. Then we readily obtain the result in Theorem \ref{pnt}.

In addition, using (\ref{rela1}) we find the following relations for the Hankel determinants:
$$
D_{2n}(t)=t^{2n\la}\tilde{D}_n\left(\frac{1}{t},-\frac{1}{2}\right)\tilde{D}_n\left(\frac{1}{t},\frac{1}{2}\right),
$$
$$
D_{2n+1}(t)=t^{(2n+1)\la}\tilde{D}_{n+1}\left(\frac{1}{t},-\frac{1}{2}\right)\tilde{D}_n\left(\frac{1}{t},\frac{1}{2}\right).
$$
However, it should be pointed out that one can not derive the differential equation (\ref{hnd}) by using the above relations and the results in \cite{BC,Chen2012}.

\section{Large $n$ Asymptotics of the Recurrence Coefficient, Sub-leading Coefficient and the Hankel Determinant}
In random matrix theory (RMT), our Hankel determinant $D_n(t)$ can be viewed as the partition function for the perturbed Gaussian unitary ensemble \cite{Mehta}:
$$
D_n(t)=\frac{1}{n!}\int_{\mathbb{R}^n}\prod_{1\leq j<k\leq n}(x_j-x_k)^2\prod_{l=1}^n \mathrm{e}^{-x_{l}^2}\left(1+t\:x_{l}^2\right)^\la dx_l.
$$
Here $x_1, x_2, \ldots, x_n$ are the eigenvalues of $n\times n$ Hermitian matrices from the ensemble, and the joint probability density function is
$$
p(x_1, x_2, \ldots, x_n)\prod_{k=1}^n dx_k=\frac{1}{n!\:D_n(t)}\prod_{1\leq j<k\leq n}(x_j-x_k)^2\prod_{l=1}^n \mathrm{e}^{-x_{l}^2}\left(1+t\:x_{l}^2\right)^\la dx_l.
$$
See also \cite{Deift,Forrester} for more information on this topic.

From Dyson's Coulomb fluid approach \cite{Dyson}, the eigenvalues (particles) can be approximated as a continuous fluid with an equilibrium density $\s(x)$ supported on $J\subset\mathbb{R}$ for sufficiently large $n$. Since our potential $\mathrm{v}(x)$ is even, it was shown by Corollary 1.12 in \cite[p. 203]{Saff} that when
$$
x \mathrm{v}'(x)=2x^2\left(1-\frac{\la t}{1+tx^2}\right)
$$
is positive and increasing on $(0,\infty)$, $J$ is a single interval denoted by $(-b,b)$. It follows that $\la t\leq 1\; (t>0)$. This is the so-called one-cut case.
Note that if $\la$ and $t$ does not satisfy the above condition, then the support $J$ could be the union of several disjoint intervals. In this section, we would like to study the large $n$ behavior of the recurrence coefficient $\bt_n(t)$, the sub-leading coefficient $\mathrm{p}(n,t)$ and the Hankel determinant $D_n(t)$ in the one-cut case.

According to \cite{ChenIsmail}, the equilibrium density $\sigma(x)$ is determined by minimizing the free energy functional
\be\label{fe1}
F[\s]:=\int_{-b}^{b}\s(x)\mathrm{v}(x)dx-\int_{-b}^{b}\int_{-b}^{b}\s(x)\ln|x-y|\s(y)dxdy
\ee
subject to
\be\label{con}
\int_{-b}^{b}\s(x)dx=n.
\ee

Then the density $\s(x)$ satisfies the integral equation
\be\label{ie}
\mathrm{v}(x)-2\int_{-b}^{b}\ln|x-y|\s(y)dy=A,\qquad x\in (-b,b),
\ee
where $A$ is the Lagrange multiplier for the constraint (\ref{con}). Note that $A$ is a constant independent of $x$ but it depends on $n$ and $t$.

The combination of (\ref{fe1}), (\ref{con}) and (\ref{ie}) gives an alternative expression of $F[\s]$:
\be\label{fe3}
F[\s]=\frac{nA}{2}+\frac{1}{2}\int_{-b}^{b}\s(x)\mathrm{v}(x)dx.
\ee
Equation (\ref{ie}) is transformed into the following singular integral equation by taking a derivative with respect to $x$,
\be\label{sie}
\mathrm{v}'(x)-2P\int_{-b}^{b}\frac{\sigma(y)}{x-y}dy=0,\qquad x\in (-b,b),
\ee
where $P$ denotes the Cauchy principal value.

The solution of (\ref{sie}) subject to the boundary condition $\sigma(-b)=\sigma(b)=0$ and $\sigma(x)\geq 0$ on $(-b,b)$ is
\be\label{sigma}
\sigma(x)=\frac{\sqrt{b^2-x^2}}{2\pi^2}P\int_{-b}^{b}\frac{\mathrm{v}'(x)-\mathrm{v}'(y)}{(x-y)\sqrt{b^2-y^2}}dy.
\ee
Then the normalization condition (\ref{con}) becomes
\be\label{sup2}
\frac{1}{2\pi}\int_{-b}^{b}\frac{x\:\mathrm{v}'(x)}{\sqrt{b^2-x^2}}dx=n.
\ee
Substituting (\ref{vp}) into (\ref{sigma}), we obtain
\be\label{den}
\sigma (x)=\frac{\sqrt{b^2-x^2} }{\pi }\left[1-\frac{\la  t}{\sqrt{1+b^2 t} \left(1+t x^2\right)}\right].
\ee
Obviously, the condition $\la t\leq 1\; (t>0)$ guarantees that the above density is nonnegative on $(-b,b)$.

Substituting (\ref{vpz}) for $\mathrm{v}'(x)$ into (\ref{sup2}) produces an equation satisfied by $b$,
\be\label{cub}
b^2-2\la+\frac{2\la}{\sqrt{1+b^2 t}}=2n.
\ee
It can be easily seen that
\be\label{cond}
b^2<2n+2\la\quad \mathrm{if}\quad  \la>0\quad  \mathrm{and}\quad  b^2>2n+2\la\quad  \mathrm{if}\quad  \la<0.
\ee
Equation (\ref{cub}) is actually a cubic equation for $b^2$, which has a unique solution under the condition (\ref{cond}). Then we find that as $n\rightarrow\infty$,
\begin{small}
\bea\label{b2}
b^2&=&2n+2\la-\frac{\sqrt{2}\:\la}{\sqrt{nt}}+\frac{\la  (2 \la  t+1)}{2 \sqrt{2}\: (nt)^{3/2}}-\frac{\la ^2}{2 n^2 t}-\frac{3 \la  (2 \la  t+1)^2}{16 \sqrt{2}\: (nt)^{5/2}}+\frac{\la ^2 (2 \la  t+1)}{2 n^3 t^2}\nonumber\\[10pt]
&+&\frac{5 \la   (8 \la ^3 t^3+4 \la ^2 t^2+6 \la  t+1)}{64 \sqrt{2}\: (nt)^{7/2}}-\frac{3 \la ^2 (2 \la  t+1)^2}{8 n^4 t^3}-\frac{35 \la  (4 \la^2  t^2-1) (4 \la ^2 t^2-8 \la  t-1)}{1024 \sqrt{2}\: (nt)^{9/2}}\nonumber\\[10pt]
&+&\frac{\la ^2 (4 \la  t+1) (2 \la ^2 t^2+2 \la  t+1)}{4 n^5 t^4}+\frac{63 \la  (2 \la  t+1)^2 (8 \la ^3 t^3-68 \la ^2 t^2+6 \la  t+1)}{4096 \sqrt{2}\: (nt)^{11/2}}+O(n^{-6}).
\eea
\end{small}
\begin{remark}
When $\la=0$, from (\ref{den}) and (\ref{cub}) we have
$$
\sigma (x)=\frac{\sqrt{b^2-x^2} }{\pi },\qquad x\in (-b,b)
$$
and $b=\sqrt{2n}$. This is the celebrated Wigner's semicircle law \cite[p. 67]{Mehta}.
\end{remark}

From (\ref{ie}) and following the similar calculations in \cite[Lemma 3]{Min2021}, we find
\be\label{ae}
A=\frac{b^2}{2}-n\ln\frac{b^2}{4}-2\la\ln\frac{1+\sqrt{1+b^2 t}}{2}.
\ee
Substituting (\ref{pt}), (\ref{den}) and (\ref{ae}) into (\ref{fe3}) and using (\ref{b2}), we obtain as $n\rightarrow\infty$,
\bea\label{fe}
F[\s]&=&-\frac{1}{2} n^2\ln n-n\la\ln n-\frac{1}{2} \la ^2 \ln n+\left(\frac{3}{4}+\frac{\ln 2}{2}\right)n^2 +n\la\left(1+\ln\frac{2}{t}\right)-\frac{2 \sqrt{2n} \la }{\sqrt{t}}\nonumber\\[10pt]
&+&\frac{\la\left(2+\la t\ln\frac{8}{t}\right)}{2t}-\frac{\la  (6 \la  t+1)}{3 \sqrt{2n}\: t^{3/2}}-\frac{\la ^2 (2 \la  t-3)}{12 n t}+\frac{\la   \left[20 \la  t (3 \la  t+1)+3\right]}{120 \sqrt{2}\:n^{3/2} t^{5/2}}\nonumber\\[10pt]
&+&\frac{\la ^2 \left(2 \la ^2 t^2-12 \la  t-3\right)}{48 n^2 t^2}+O(n^{-5/2}).
\eea
\begin{remark}
Inserting (\ref{b2}) into (\ref{ae}), we have as $n\rightarrow\infty$,
\bea\label{A}
A&=&n\left(1+\ln\frac{2}{n}\right)+\la\ln\frac{2}{nt}-\frac{\sqrt{2}\la}{\sqrt{nt}}-\frac{\la^2}{2n}+\frac{\la(1+6\la t)}{6\sqrt{2}\:(nt)^{3/2}}+\frac{\la ^2 (2 \la  t-3)}{12 n^2 t}\nonumber\\[10pt]
&-&\frac{ \la  \left[20 \la  t (3 \la  t+1)+3\right]}{80 \sqrt{2}\: (nt)^{5/2}}+\frac{\la ^2 \left[3-2 \la  t (\la  t-6)\right]}{24 n^3 t^2}+O(n^{-7/2}).
\eea
It is easy to check that (\ref{fe}) and (\ref{A}) satisfy the relation \cite[(2.14)]{ChenIsmail}
$$
\frac{\partial F}{\partial n}=A.
$$
\end{remark}

According to formula (2.27) in \cite{ChenIsmail}, we have
$$
\bt_n=\frac{b^2}{4}\left(1+O\left(\frac{\partial^{4}F}{\partial n^{4}}\right)\right),\qquad n\rightarrow\infty.
$$
In view of (\ref{b2}) and (\ref{fe}), it is obvious to see that $\bt_n$ has the large $n$ expansion of the form
\be\label{exp}
\bt_n=a_{-2}n+a_{-1}\sqrt{n}+a_0+\sum_{k=1}^{\infty}\frac{a_{k}}{n^{k/2}}, \qquad\qquad n\rightarrow\infty,
\ee
where
$$
a_{-2}=\frac{1}{2}
$$
and $a_{k},\; k=-1, 0, 1, \ldots$ are the expansion coefficients to be determined. By using the second-order difference equation satisfied by $\bt_n$, we obtain the following theorem.
\begin{theorem}\label{thm}
The recurrence coefficient $\bt_n$ has the following large $n$ expansion:
\be\label{bte}
\bt_n=\frac{n+\la}{2}+\sum_{k=1}^{\infty}\frac{a_{k}}{n^{k/2}}, \qquad\qquad n\rightarrow\infty,
\ee
where the first few terms of expansion coefficients are
\bea
&&a_{1}=-\frac{\la }{2 \sqrt{2 t}},\qquad\qquad a_{2}=0,\qquad\qquad a_{3}=\frac{\la  (2 \la  t+1)}{8 \sqrt{2}\: t^{3/2}},\qquad\qquad a_{4}=-\frac{\la ^2}{8 t},\nonumber\\[8pt]
&&a_{5}=-\frac{\la  \left[(12 \la ^2-5) t^2+12 \la  t+3\right]}{64 \sqrt{2}\: t^{5/2}},\qquad\qquad a_{6}=\frac{\la ^2(2 \la  t+1)}{8 t^2},\nonumber\\[8pt]
&&a_{7}=\frac{5 \la  \left[2 \la  (4 \la ^2-5) t^3+(4 \la ^2-7) t^2+6 \la  t+1\right]}{256 \sqrt{2}\: t^{7/2}}.\nonumber
\eea
\end{theorem}

\begin{proof}
Substituting (\ref{exp}) into the difference equation (\ref{beta}), we have an expression of the following form by letting $n\rightarrow\infty$:
$$
e_{-6}n^3+e_{-5}n^{5/2}+e_{-4}n^{2}+\sum_{k=-3}^{\infty}\frac{e_{k}}{n^{k/2}}=0,
$$
where all the coefficients of powers of $n$, $e_{k}$, are explicitly known and should be equal to $0$.
The equation $e_{-6}=0$ is
$$
t a_{-2}(2-4a_{-2})^2=0,
$$
which holds from the fact that $a_{-2}=\frac{1}{2}$.

Setting $a_{-2}=\frac{1}{2}$ leads to $e_{-5}$ vanishing identically. The equation $e_{-4}=0$ then gives rise to
$$
8t\: a_{-1}^{2}=0.
$$
Since $t>0$, we have
$$
a_{-1}=0.
$$
With $a_{-2}=\frac{1}{2}$ and $a_{-1}=0$, we find that $e_{-3}$ vanishes identically. The equation $e_{-2}=0$ gives us
$$
2t(2a_0-\la)^2=0,
$$
and we get
$$
a_0=\frac{\la}{2}.
$$
Proceeding with this procedure, we can easily obtain the higher order coefficients $a_{1},\; a_{2},\; a_{3},\ldots.$
The proof is complete.
\end{proof}

\begin{theorem}\label{thm1}
The sub-leading coefficient $\mathrm{p}(n,t)$ has the following expansion as $n\rightarrow\infty$:
\be\label{pnte}
\mathrm{p}(n,t)=-\frac{n^2}{4}+\frac{(1-2 \la)n}{4}+\frac{\la \sqrt{n}}{\sqrt{2 t}}+\frac{\la\left[(1-\la)t-2\right]}{4 t}+\sum_{k=1}^{\infty}\frac{b_{k}}{n^{k/2}},
\ee
where the first few terms of expansion coefficients are
\bea
&&b_{1}=\frac{\la\left[(2\la-1)t+1\right]}{4 \sqrt{2}\: t^{3/2}},\qquad\qquad\qquad\qquad\qquad\;\; b_{2}=-\frac{\la ^2}{8 t},\nonumber\\[8pt]
&&b_{3}=\frac{\la\left[(1+4\la-4\la^2)t^2+2(1-2\la)t-1\right]}{32 \sqrt{2}\: t^{5/2}},\qquad b_{4}=\frac{\la^2\left[(2\la-1)t+1\right]}{16 t^2},\nonumber\\[8pt]
&&b_{5}=\frac{\la  \left[(2 \la -1) (4 \la ^2-4 \la -5) t^3+(4 \la ^2-12 \la -5) t^2+3(2 \la -1) t+1\right]}{128 \sqrt{2}\: t^{7/2}}.\nonumber
\eea
\end{theorem}
\begin{proof}
Using (\ref{rne}) and (\ref{Rne}) to eliminate $r_n(t)$ and $R_n(t)$ in (\ref{iden}), we find $\mathrm{p}(n,t)$ can be expressed in terms of $\bt_n$ and $\bt_{n\pm1}$ as follows:
$$
\mathrm{p}(n,t)=\frac{n+\bt_n(2 n t+2 \lambda  t+t-2-2 t \bt_{n-1}-2t\bt_n-2 t \bt_{n+1})}{2 t}.
$$
Substituting (\ref{bte}) into the above and taking a large $n$ limit, we obtain the desired result.
\end{proof}
\begin{theorem}
The quantity $H_n(t)=2t^2\frac{d}{dt}\ln D_n(t)-2n\la t$ has the following asymptotic expansion as $n\rightarrow\infty$:
\bea\label{snt}
H_n(t)&=&-2 \la\sqrt{2nt}+\la  (\la  t+2)-\frac{\la   (2 \la  t+1)}{\sqrt{2nt}}+\frac{\la ^2}{2 n}+\frac{\la  \left[(4 \la ^2-1) t^2+4 \la  t+1\right]}{8 \sqrt{2}\: (nt)^{3/2}}\nonumber\\
&-&\frac{\la ^2 (2 \la  t+1)}{4 n^2 t}-\frac{ \la  \left[2\la(4 \la ^2-3) t^3+(4 \la ^2-5) t^2+6 \la  t+1\right]}{32 \sqrt{2}\: (nt)^{5/2}}+O(n^{-3}).
\eea
\end{theorem}
\begin{proof}
Inserting (\ref{be1}) into (\ref{p3}), we have
\be\label{hpp}
H_n(t)=-2t\:\mathrm{p}(n,t)-2t\:\mathrm{p}(n+1,t)-nt(n+2\la).
\ee
Substituting (\ref{pnte}) into (\ref{hpp}) and letting $n\rightarrow\infty$, we obtain (\ref{snt}).
\end{proof}
\begin{theorem}\label{thm2}
The Hankel determinant $D_n(t)$ has the following asymptotic expansion as $n\rightarrow\infty$:
\bea\label{dnta}
\ln D_n(t)&=&\frac{1}{2} n^2 \ln n+n\la  \ln n-\left(\frac{1}{12}-\frac{\la ^2}{2}\right) \ln n-n^2 \left(\frac{3}{4}+\frac{\ln 2}{2}\right)+n(\la\ln t+\tilde{c}_{1}(\la))\nonumber\\[10pt]
&+&\frac{2\la\sqrt{2n}}{\sqrt{t}}-\frac{\la}{t}+\frac{\la^2}{2}\ln t+\tilde{c}_0(\la)+\frac{\la  (6 \la  t+1)}{3 \sqrt{2n}\: t^{3/2}}-\frac{\la  \left[3 \la+(1 -2 \la ^2) t\right]}{12 n t}\nonumber\\[10pt]
&-&\frac{\la  \left[15 (4 \la ^2-1) t^2+20 \la  t+3\right]}{120 \sqrt{2}\: n^{3/2} t^{5/2}}+O(n^{-2}),
\eea
where $\tilde{c}_0(\la)$ and $\tilde{c}_{1}(\la)$ are constants depending on $\la$ only.
\end{theorem}

\begin{proof}
We define
$$
F_n(t):=-\ln D_n(t)
$$
to be the ``free energy''. Chen and Ismail \cite{ChenIsmail} showed that $F_n(t)$ is approximated by the free energy $F[\s]$ in (\ref{fe1}) for sufficiently large $n$. It was also pointed out in \cite{ChenIsmail1998} that the approximation is very accurate and effective.
Taking account of (\ref{fe}),
we have the following large $n$ expansion form for $F_n(t)$:
\be\label{fna}
F_n(t)=c_{5}(t,\la)n^2\ln n+c_4(t,\la) n \ln n+c_3(t,\la) \ln n+\sum_{j=-\infty}^{4}c_{j/2}(t,\la)n^{j/2}.
\ee

From (\ref{re}) we find
\be\label{dc}
-\ln\bt_n=F_{n+1}(t)+F_{n-1}(t)-2F_n(t).
\ee
Substituting (\ref{bte}) and (\ref{fna}) into equation (\ref{dc}) and letting $n\rightarrow\infty$, we obtain the asymptotic expansion for $F_n(t)$ by equating coefficients of powers of $n$:
\bea
F_n(t)&=&-\frac{1}{2} n^2 \ln n-n\la  \ln n+\left(\frac{1}{12}-\frac{\la ^2}{2}\right) \ln n+n^2 \left(\frac{3}{4}+\frac{\ln 2}{2}\right)+c_1(t,\la)n
\nonumber\\[6pt]
&-&\frac{2\la\sqrt{2n}}{\sqrt{t}}+c_0(t,\la)-\frac{\la  (6 \la  t+1)}{3 \sqrt{2n}\: t^{3/2}}+\frac{\la  \left[3 \la+(1 -2 \la ^2) t\right]}{12 n t}\nonumber\\[6pt]
&+&\frac{\la  \left[15 (4 \la ^2-1) t^2+20 \la  t+3\right]}{120 \sqrt{2}\: n^{3/2} t^{5/2}}+O(n^{-2}),\no
\eea
where $c_{1}(t,\la)$ and $c_0(t,\la)$ are undetermined coefficients independent of $n$.
It follows that
\bea\label{dn}
\ln D_n(t)&=&\frac{1}{2} n^2 \ln n+n\la  \ln n-\left(\frac{1}{12}-\frac{\la ^2}{2}\right) \ln n-n^2 \left(\frac{3}{4}+\frac{\ln 2}{2}\right)-c_1(t,\la)n\nonumber\\[6pt]
&+&\frac{2\la\sqrt{2n}}{\sqrt{t}}-c_0(t,\la)+\frac{\la  (6 \la  t+1)}{3 \sqrt{2n}\: t^{3/2}}-\frac{\la  \left[3 \la+(1 -2 \la ^2) t\right]}{12 n t}\nonumber\\[6pt]
&-&\frac{\la  \left[15 (4 \la ^2-1) t^2+20 \la  t+3\right]}{120 \sqrt{2}\: n^{3/2} t^{5/2}}+O(n^{-2}).
\eea

Taking a derivative with respect to $t$ in (\ref{dn}) and substituting it into (\ref{hd}), we find
\bea\label{com}
H_n(t)&=&-2n t \left(\la +t\frac{d}{dt} c_1(t,\la)\right)-2\la \sqrt{2nt}-2 t^2\frac{d}{dt} c_0(t,\la)-\frac{\la  (2 \la  t+1)}{\sqrt{2nt} }\nonumber\\[6pt]
&+&\frac{\la ^2}{2 n}+\frac{\la  \left[(4 \la ^2-1) t^2+4 \la  t+1\right]}{8 \sqrt{2}\:  (nt)^{3/2}}+O(n^{-2}).
\eea
Making a comparison between (\ref{com}) and (\ref{snt}) gives
$$
\la +t\frac{d}{dt} c_1(t,\la)=0,
$$
$$
-2 t^2\frac{d}{dt} c_0(t,\la)=\la  (\la  t+2).
$$
We then have
\be\label{c1}
c_{1}(t,\la)=-\la\ln t-\tilde{c}_{1}(\la),
\ee
\be\label{c0}
c_0(t,\la)=\frac{\la}{t}-\frac{\la^2}{2}\ln t-\tilde{c}_0(\la),
\ee
where $\tilde{c}_0(\la)$ and $\tilde{c}_{1}(\la)$ are constants depending on $\la$ only. The theorem follows from
inserting (\ref{c1}) and (\ref{c0}) into (\ref{dn}).
\end{proof}

If $\la=0$, our weight is reduced to the standard Gaussian weight $w(x)=\mathrm{e}^{-x^2},\;x\in \mathbb{R}$. In this case, it is well known that \cite{Szego}
\be\label{btg}
\bt_n=\frac{n}{2},
\ee
\be\label{png}
\mathrm{p}(n)=-\frac{n^2}{4}+\frac{n}{4},
\ee
\be\label{hng}
h_n=\frac{n!}{2^n}\sqrt{\pi}.
\ee
Obviously, (\ref{btg}) and (\ref{png}) coincide with the results in Theorem \ref{thm} and Theorem \ref{thm1} respectively when $\la=0$.

From (\ref{hng}) we have the Hankel determinant for the standard Gaussian weight
$$
D_n=(2\pi)^{\frac{n}{2}}2^{-\frac{n^2}{2}}G(n+1),
$$
where $G(\cdot)$ is the Barnes $G$-function and satisfies the functional equation \cite{Barnes,Choi,Voros}
$$
G(z+1)=\Gamma(z)G(z),\qquad\qquad G(1):=1.
$$
It follows that
$$
\ln D_n=-\frac{n^2}{2}\ln 2+\frac{n}{2}\ln(2\pi)+\ln G(n+1).
$$

By using the asymptotic formula of Barnes $G$-function \cite[p. 285]{Barnes}
$$
\ln G(z+1)=z^2\left(\frac{\ln z}{2}-\frac{3}{4}\right)+\frac{z}{2}\ln (2\pi)-\frac{\ln z}{12}+\zeta'(-1)+O(z^{-2}),\qquad z\rightarrow+\infty,
$$
where $\zeta(\cdot)$ is the Riemann zeta function, we obtain
$$
\ln D_n=\frac{n^2}{2}\ln n-\left(\frac{\ln 2}{2}+\frac{3}{4}\right)n^2+n\ln(2\pi)-\frac{1}{12}\ln n+\zeta'(-1)+O(n^{-2}),\qquad n\rightarrow\infty.
$$
Setting $\la=0$ in (\ref{dnta}) and comparing it with the above, we have
$$
\tilde{c}_{1}(0)=\ln(2\pi),
$$
$$
\tilde{c}_{0}(0)=\zeta'(-1).
$$
However, the constants $\tilde{c}_0(\la)$ and $\tilde{c}_{1}(\la)$ in Theorem \ref{thm2} can still not be determined completely with our method.
\begin{remark}
The large $n$ asymptotic expansions of $\bt_n(t)$, $\mathrm{p}(n,t)$, $H_n(t)$ and $D_n(t)$ obtained in this section are only valid when $\la$ and $t$ satisfy the condition $\la t\leq 1,\; t>0$.
\end{remark}

\section*{Acknowledgments}
The work of Chao Min was partially supported by the National Natural Science Foundation of China under grant number 12001212, by the Fundamental Research Funds for the Central Universities under grant number ZQN-902 and by the Scientific Research Funds of Huaqiao University under grant number 17BS402. The work of
Yang Chen was partially supported by the Macau Science and Technology Development Fund under grant number FDCT 0079/2020/A2.


\end{document}